\documentclass[conference]{IEEEtran}
\IEEEoverridecommandlockouts

\usepackage{cite}
\usepackage{amsmath,amssymb,amsfonts,bm,amsthm}
\usepackage{algorithm}
\usepackage{algorithmic}
\usepackage{array}
\usepackage{multirow}
\usepackage[caption=false,font=footnotesize]{subfig}
\usepackage{textcomp}
\usepackage{stfloats}
\usepackage{url}
\usepackage{graphicx}
\usepackage{booktabs}
\usepackage{xcolor}
\usepackage{balance}
\usepackage{enumitem}
\usepackage{pifont}
\allowdisplaybreaks[4]

\newtheorem{proposition}{Proposition}
\newcommand{\method}{RC-BFM}

\begin{document}

\title{Low-Latency Generative Semantic Communication via Channel-Realization Flow Matching
    \thanks{This work was supported in part by the National Key Research and
        Development Program of China under Grant No. 2023YFB2904300;
        and in part by the National Natural Science Foundation of China under Grant No. 62595731, No. 62293484, and No. 62325107; and in part by the Program of Jiangsu Province under Grant NTACT-2024-Z-001.
    }
}

\author{
    \IEEEauthorblockN{Fan Gao$^{1,2}$, Youzheng Wang$^{1,2}$, Zhijin Qin$^{1,2}$, and Feifei Gao$^{1,3}$}
    \IEEEauthorblockA{$^1$Beijing National Research Center for Information Science and Technology (BNRist), Tsinghua University, Beijing, China \\
        $^2$Department of Electronic Engineering, $^3$Department of Automation, Tsinghua University, Beijing, China \\
        gaof23@mails.tsinghua.edu.cn, \{yzhwang, qinzhijin\}@tsinghua.edu.cn, feifeigao@ieee.org}
}

\maketitle

\begin{abstract}
    Generative semantic communication receivers deliver high perceptual quality but suffer from prohibitive decoding latency. This bottleneck arises because diffusion receivers rely on stochastic iterative decoding, while existing flow matching receivers employ independent endpoint coupling that ignores the physical source--channel link, yielding unnecessarily long and curved sampling trajectories. In this paper, we reformulate receiver-side recovery as a realization-coupled bridge flow matching problem under explicit bandwidth and power constraints. Specifically, we propose Realization-Coupled Bridge Flow Matching (RC-BFM), where the decoder initializes from a channel-induced semantic state rather than isotropic noise. Crucially, training pairs are linked via a realization-coupled entropic optimal transport (RC-OT) plan that preserves the physical channel realization of each transmission while maintaining robustness to stochastic fading. Furthermore, we identify independent coupling as the fundamental source of a conditional train--test distribution shift in conditional flow matching-based receivers, and derive an end-to-end distortion bound whose discretization error decays as \(O(K^{-2})\). Experiments on CIFAR-10 and FFHQ-64\(\times\)64 over AWGN and Rayleigh fading channels demonstrate that RC-BFM achieves a superior fidelity--perception trade-off, reducing decoding latency by over 10\(\times\) compared to diffusion-based receivers.
\end{abstract}


\section{Introduction}

% 6G->semcom->deepjscc
Sixth-generation (6G) wireless networks target latency-critical visual applications such as extended reality, autonomous perception, and digital twins~\cite{ITU_M2160_2023,you2025next}. Semantic communication addresses the resulting bandwidth--latency tension by transmitting the underlying meaning of data rather than raw bit streams~\cite{qin2025generative,gunduz2023beyond}. DeepJSCC~\cite{bourtsoulatze2019deep} maps source images directly to channel symbols via an end-to-end encoder-decoder, avoiding the cliff effect of source-channel separation. But optimizing pixel-level distortion metrics only produces blurry, perceptually unrealistic reconstructions that lack fine-grained textures and semantic consistency under low bandwidths or severe channel impairments~\cite{erdemir2023generative}.

Generative receivers sidestep the fidelity--realism tradeoff by leveraging deep generative priors~\cite{qin2025generative}. Diffusion-based designs~\cite{wang2025diffcom,zhang2025sgdjscc} recover perceptually realistic detail through iterative posterior sampling from isotropic Gaussian noise. However, the reverse chain demands tens to hundreds of neural function evaluations (NFEs), which is incompatible with real-time wireless budgets. The LTT receiver~\cite{fu2026ltt} applies Flow Matching (FM)~\cite{lipman2023flow,liu2023rectified} to semantic communication, cutting NFEs by an order of magnitude while retaining perceptual quality, which straightens the ODE integration trajectories by constructing an interpolation-based velocity field to accelerate receiver-side recovery.

Despite these advances, a fundamental mismatch persists in all existing FM receivers: training pipelines pair the source and target endpoints independently. While harmless in unconditional image generation, this independent product coupling ignores the physical channel realization~\cite{tse2005fundamentals} that links the clean source to the corrupted observation in a communication system. Replacing this physically grounded joint distribution with an independent random coupling has two severe consequences. First, the channel-conditioned source marginal seen by the bridge ODE during training differs from the physical source law at deployment, creating a conditional train-test mismatch that degrades receiver recovery quality. Second, the resulting ODE trajectories are empirically longer and more curved than necessary, limiting receiver decoding speed.

In this paper, we leverage optimal transport (OT)~\cite{cuturi2013sinkhorn} and its entropic extension~\cite{DeBortoli2021DSB} to address both issues. We propose Realization-Coupled Bridge FM (RC-BFM), which explicitly preserves the physical source--channel pairing. Our contributions are summarized as follows:
\begin{enumerate}[label=\arabic*), leftmargin=*]
    \item \textbf{Flow matching based receiver framework:} We propose a realization-coupled bridge FM framework for generative semantic communication. By formulating receiver-side reconstruction as a bridge flow matching decoding from a channel-induced state to the clean image manifold, it ensures the ODE starts from an informative prior rather than isotropic noise.
    \item \textbf{Realization-coupled entropic OT modeling:} We model the physical source--channel link as a realization-coupled pairing problem, and propose a realization-coupled optimal transport (RC-OT) method to solve the conditional train-test mismatch. The coupling concentrates mass on pairs that share the same source image and channel realization, while an entropic regularization provides distributional robustness against stochastic fading.
    \item \textbf{Theoretical and empirical validation:} We identify independent coupling as the source of a conditional train--test distribution shift in CFM-based receivers, and derive a concise distortion decomposition bounded by an $O(K^{-2})$ discretization error term. Experiments on CIFAR-10 and FFHQ-64\(\times\)64 over AWGN and Rayleigh fading show that RC-BFM achieves a superior fidelity--perception balance with more than 10\(\times\) lower latency than diffusion-based receivers.
\end{enumerate}

\section{System Model}

\begin{figure}[t]
    \centering
    \includegraphics[width=\linewidth]{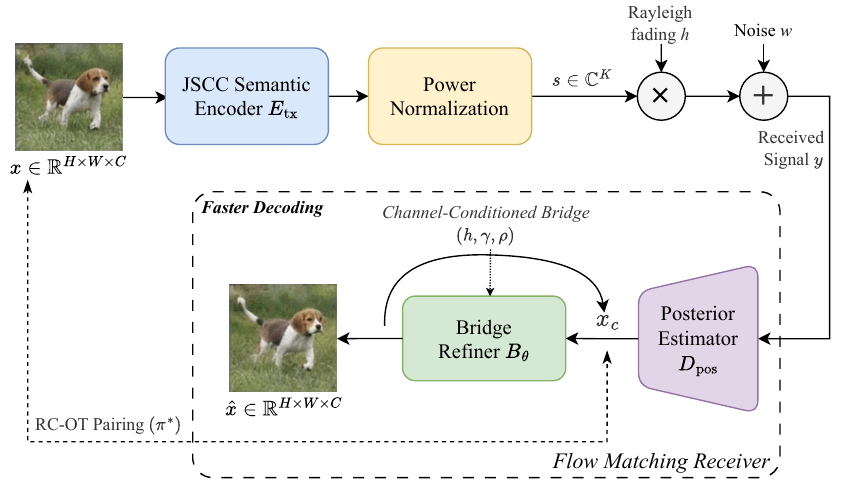}
    \caption{End-to-end RC-BFM system. RC-OT pairing $\pi^*$ enforces a physically grounded source--observation coupling at training time.}
    \label{fig:system_model}
\end{figure}

Fig.~\ref{fig:system_model} illustrates the overall system. Let $x\in\mathbb{R}^{H\times W\times 3}$ denote a source image. A trainable JSCC semantic encoder $E_{\mathrm{tx}}$ maps $x$ to $K$ complex channel symbols $s\in\mathbb{C}^{K}$ subject to the power constraint $\mathbb{E}[\|s\|^2]\le P$, with channel bandwidth ratio (CBR) $\rho = K/(H\times W\times 3)$.

We consider both AWGN and Rayleigh fading channels. The received signal $y$ is
\begin{equation}
    y = h s + w, \quad h\sim\mathcal{CN}(0,1), \quad w\sim\mathcal{CN}(0,\sigma_w^2),
    \label{eq:rayleigh_ch}
\end{equation}
where the SNR is $\gamma=P/\sigma_w^2$ and the channel condition $\xi=(h,\gamma)$ collects all available channel-side information (CSI).

We model the receiver as a bridge flow matching receiver. A posterior estimator first produces a coarse estimate $x_c=D_{\mathrm{pos}}(y,\xi)$, which we term the \emph{channel-induced semantic state}. A neural bridge refiner $B_\theta$ then evolves $x_c$ along a probability flow toward the clean image via the channel-conditioned ODE
\begin{equation}
    \hat{x} = x_c + \int_0^1 v_\theta(x_t, t, \xi)\,dt, \quad x_0 = x_c,
    \label{eq:fm_receiver}
\end{equation}
where $v_\theta$ is a learned velocity field. At training time, $v_\theta$ is fitted by a flow matching loss; at inference, the receiver decodes by directly integrating~\eqref{eq:fm_receiver} with a $K$-step Euler solver from $x_0=x_c$. Section~III details both procedures.

The end-to-end design objective minimizes joint fidelity and perceptual distortion subject to bandwidth, power, and latency constraints:
\begin{equation}
    \begin{aligned}
        & \min_{E_{\mathrm{tx}},D_{\mathrm{pos}},B_\theta}\quad
         \mathbb E\big[d_{\mathrm{fid}}(x,\hat x)+\lambda_{\mathrm{perc}}d_{\mathrm{perc}}(x,\hat x)\big] \\
        & \text{s.t.}\quad
          \rho\le \rho_{\max},\quad \mathbb E[\|\bm{s}\|^2]\le P,\quad \mathrm{NFE}(B_\theta)\le K_{\max},
    \end{aligned}
    \label{eq:overall_obj}
\end{equation}
where $d_{\mathrm{fid}}(x,\hat x)=\|x-\hat x\|_2^2$ measures $L_2$ fidelity, $d_{\mathrm{perc}}$ denotes the LPIPS perceptual distance~\cite{zhang2018unreasonable}, and $K_{\max}$ caps the per-image NFE that directly controls decoding latency.

\section{Proposed Realization-Coupled Bridge Flow Matching}

\subsection{Bridge Flow Matching}

Rather than initiating the generative process from isotropic Gaussian noise, our decoder constructs a probability bridge directly from the channel-induced state $x_c$ to the clean image $x$. We define the bridge mean $\mu_t$, variance $s_t^2$, and stochastic path $x_t$ as follows:
\begin{equation}
    \mu_t = (1{-}t)x_c + tx, \quad s_t^2 = \sigma^2 t(1{-}t),
    \label{eq:bridge_moments}
\end{equation}
\begin{equation}
    x_t = \mu_t + \sigma\sqrt{t(1{-}t)}\,\epsilon, \quad \epsilon\sim\mathcal{N}(0,I),
    \label{eq:bridge_path}
\end{equation}
where $t\in[0,1]$ denotes the continuous bridge time and $\sigma>0$ controls the degree of bridge stochasticity. The corresponding conditional velocity field is given by
\begin{equation}
    u_t(x_t\mid x_c,x) = x - x_c + \frac{1{-}2t}{2t(1{-}t)}\bigl(x_t - \mu_t\bigr).
    \label{eq:cond_velocity}
\end{equation}
The first term, $x - x_c$, represents the constant endpoint displacement, while the second term compensates for the Gaussian perturbation along the bridge. We train a neural velocity field $v_\theta(x_t,t,\xi)$ to approximate the marginal probability-flow field using the conditional FM (CFM) objective~\cite{lipman2023flow}:
\begin{equation}
    \mathcal{L}_{\mathrm{bfm}} = \mathbb{E}\left[\left\|v_\theta(x_t,t,\xi) - u_t(x_t\mid x_c,x)\right\|_2^2\right].
    \label{eq:cfm_loss}
\end{equation}

During inference, we obtain the reconstructed image $\hat{x}$ by integrating the empirical ODE $\dot{x}_t = v_\theta(x_t,t,\xi)$ from $x_0=x_c$ to $t=1$. We employ the Euler solver with $K$ steps of size $h=1/K$:
\begin{equation}
    x_{k+1} = x_k + h\,v_\theta(x_k, t_k, \xi), \quad t_k = kh.
    \label{eq:euler}
\end{equation}
Because the posterior initialization $x_c$ is close to the target $x$, the integration interval is significantly shorter than the noise-to-data path traversed by diffusion models. The Euler discretization error scales as $O(M_v/K)$, where $M_v$ bounds the curvature of the velocity field along the trajectory. By combining the channel-induced semantic initialization with the realization-coupled optimal transport (RC-OT) coupling introduced below, we effectively minimize $M_v$, thereby achieving high fidelity with very few solver steps.

\begin{proposition}[Low-NFE Euler Error]
    Assume that $v_\theta(\cdot,t,\xi)$ is $L_v$-Lipschitz in state, and that the continuous trajectory $x(t)$ satisfies the curvature bound:
    \begin{equation}
        \left\|\partial_t v_\theta(x(t),t,\xi)+J_xv_\theta(x(t),t,\xi)v_\theta(x(t),t,\xi)\right\|\le M_v
    \end{equation}
    for all $t\in[0,1]$. The Euler iterates $x_k$ from~\eqref{eq:euler} then satisfy:
    \begin{equation}
        \max_{0\le k\le K} \|x(t_k)-x_k\|_2 \le \frac{e^{L_v}-1}{2L_v}\frac{M_v}{K}.
        \label{eq:euler_bound_conf}
    \end{equation}
\end{proposition}
\begin{proof}
    A one-step Taylor expansion yields $x(t_{k+1})=x(t_k)+h v_\theta(x(t_k),t_k,\xi)+\tfrac{h^2}{2}r_k$ with $\|r_k\|\le M_v$. Defining the global error $e_k=x(t_k)-x_k$, we obtain the recursion $\|e_{k+1}\|\le (1+hL_v)\|e_k\|+\tfrac{h^2}{2}M_v$. Iterating from $e_0=0$ and applying $(1+hL_v)^K\le e^{L_v}$ with $h=1/K$ yields~\eqref{eq:euler_bound_conf}.
\end{proof}

This proposition directly yields a communication-oriented distortion decomposition. Let $x_1^\star$ denote the exact terminal state of the continuous ODE at $t=1$, and define the end-to-end objective:
\begin{equation}
    \mathcal D_K \triangleq \mathbb E\big[d_{\mathrm{fid}}(x,\hat x_K)+\lambda_{\mathrm{perc}}d_{\mathrm{perc}}(x,\hat x_K)\big].
    \label{eq:dist_obj_conf}
\end{equation}
Assuming the perceptual extractor is Lipschitz continuous such that $\mathcal D_K\le C_{\mathrm{lat}}\mathbb E\|x-\hat x_K\|_2^2$ for some constant $C_{\mathrm{lat}}>0$, we can bound the total distortion as:
\begin{equation}
    \begin{aligned}
        \mathcal D_K
         & \le 3C_{\mathrm{lat}}\,\mathbb E\|x-x_c\|_2^2 + 3C_{\mathrm{lat}}\,\mathbb E\|x_c-x_1^\star\|_2^2 \\
         & \quad + 3C_{\mathrm{lat}}\left(\frac{e^{L_v}-1}{2L_v}\frac{M_v}{K}\right)^2.
    \end{aligned}
    \label{eq:main_bound_conf}
\end{equation}
Equation~\eqref{eq:main_bound_conf} separates the distortion into three distinct sources: (1) the posterior-estimation error from the estimator $D_{\mathrm{pos}}$, (2) the intrinsic flow-matching residual of the learned vector field, and (3) the latency penalty induced by Euler discretization. This decomposition highlights why realization-coupled pairing is essential: by shortening and straightening the bridge paths, RC-OT reduces the effective curvature $M_v$, suppressing the $O(K^{-2})$ latency penalty for a given budget $K$.

\subsection{Coupling Mismatch in Semantic Communication}

The CFM objective in~\eqref{eq:cfm_loss} requires sampling endpoint pairs $(x_c,x)$ from a joint coupling $\pi(x_c,x)$. Standard FM techniques~\cite{lipman2023flow,liu2023rectified,fu2026ltt} form this coupling by drawing samples independently from the source and target marginals. However, in communication systems, each corrupted observation $x_c$ and its corresponding clean target $x$ are physically linked by the specific channel realization $\xi$ that generated them. This physical link defines a paired conditional law:
\begin{equation}
    q_{\mathrm{pair}}(x_c,x\mid\xi) = q(x\mid\xi)\,q(x_c\mid x,\xi).
    \label{eq:qpair}
\end{equation}
Whenever $x_c$ depends on $x$ through the channel, this true paired distribution differs strictly from the independent product distribution $q(x_c\mid\xi)\,q(x\mid\xi)$.

Integrating $x$ out of an independent coupling $\tilde\pi(x_c,x\mid\xi)=q(x_c\mid\xi)q(x\mid\xi)$ yields the correct source marginal $q(x_c\mid\xi)$ only if $x_c$ is entirely independent of $x$. For any coupling that ignores the shared channel realization, the induced source marginal $\tilde q(x_c\mid\xi)=\int\tilde\pi(x_c,x\mid\xi)\,dx$ diverges from the physical marginal $q(x_c\mid\xi)$ that the ODE encounters at deployment. The bridge therefore learns to integrate from an incorrect initial distribution. This mismatch between the training and deployment marginals constitutes a conditional covariate shift~\cite{cheng2025curse}.

Restricting every pair to share the same source sample and channel realization restores the correct marginal by construction. The same-realization coupling $\pi^{\mathrm{pair}}(x_c,x\mid\xi)=q(x\mid\xi)\,q(x_c\mid x,\xi)$ trivially satisfies $\int\pi^{\mathrm{pair}}(x_c,x\mid\xi)\,dx = q(x_c\mid\xi)$. Geometrically, this constraint shortens and straightens the bridge paths as illustrated in Fig.~\ref{fig:path_geometry}. This geometry reduces the velocity-field curvature $M_v$ that enters Proposition~1, directly suppressing the discretization error at a fixed NFE budget. We therefore need a coupling that concentrates mass on physical pairs while remaining differentiable and tractable on minibatches.

\begin{figure}[t]
    \centering
    \includegraphics[width=\linewidth]{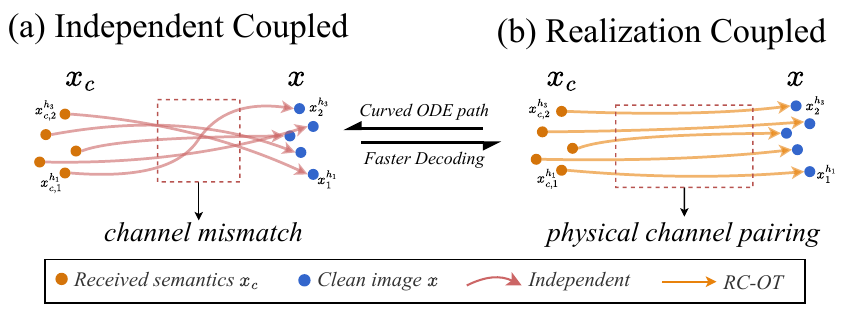}
    \caption{Coupling comparison. (a)~Independent coupling pairs $(x_c,x)$ regardless of the channel realization $h_i$, causing a channel mismatch and crossed, highly curved ODE paths. (b)~RC-OT coupling concentrates mass on physical pairs that share the same source and channel realization, producing straight, non-crossing trajectories that translate into faster decoding.}
    \label{fig:path_geometry}
\end{figure}

\subsection{Realization-Coupled Entropic OT}

Motivated by the mismatch argument above, we construct an optimal transport coupling that preserves the physical source--receiver correspondence while maintaining tractability over training minibatches. Let $C_{ij}=\|x_{c,i}-x_j\|_2^2$ denote the pairwise cost matrix over a minibatch of size $B$. We define a realization-coupled reference measure $R=\alpha I/B + (1{-}\alpha)\bm{a}\bm{b}^\top$, where $\bm{a},\bm{b}\in\mathbb{R}^B$ are uniform marginals and $\alpha\in[0,1]$ controls the concentration on same-realization pairs.

We obtain the realization-coupled entropic OT plan by solving the regularized optimal transport problem~\cite{akshay2025unified}:
\begin{equation}
    \pi^* = \arg\min_{\pi\in\Pi(\bm{a},\bm{b})} \langle C,\pi\rangle + \varepsilon\,\mathrm{KL}(\pi\,\Vert\, R),
    \label{eq:ot_sb}
\end{equation}
where $\varepsilon>0$ dictates the strength of the entropic regularization. Using a generalized Sinkhorn factorization~\cite{cuturi2013sinkhorn}, the solution takes the form:
\begin{equation}
    \pi^* = \mathrm{Diag}(\bm{u})\,\bm{K}\,\mathrm{Diag}(\bm{v}), \quad K_{ij} = R_{ij}\exp\!\left(-\frac{C_{ij}}{\varepsilon}\right),
    \label{eq:sinkhorn_form}
\end{equation}
where $\bm{u}$ and $\bm{v}$ are strictly positive scaling vectors computed via alternating row and column normalization.

As depicted in Fig.~\ref{fig:path_geometry}, this coupling strategy primarily concentrates probability mass along the diagonal (enforcing physical pairing). Simultaneously, the entropic regularization allows the optimizer to smoothly reassign pairs when the transport cost $C_{ij}$ is overwhelmingly high due to severe channel noise. Because the independent plan $\pi_{\mathrm{ind}}=\bm{a}\bm{b}^\top$ is always feasible for~\eqref{eq:ot_sb}, our optimal plan $\pi^*$ is mathematically guaranteed to achieve a regularized cost no worse than the independent coupling widely used in prior work.

\subsection{Training and Inference}

We train the proposed RC-BFM system in two stages. Stage~1 performs channel-free generative pre-training. We deliberately omit the semantic channel encoder $E_{\mathrm{tx}}$ and pre-train only the bridge refiner $B_\theta$ on clean images to establish a robust generative prior. Stage~2 incorporates the physical channel. We introduce the semantic encoder $E_{\mathrm{tx}}$ and the posterior estimator $D_{\mathrm{pos}}$, and jointly train the entire end-to-end system ($E_{\mathrm{tx}}$, $D_{\mathrm{pos}}$, and $B_\theta$) under noise and fading impairments. The endpoint pairs $(x_c, x)$ are drawn directly from the computed RC-OT coupling $\pi^*$. We optimize the total end-to-end objective:
\begin{equation}
    \mathcal{L} = \lambda_{\mathrm{post}}\mathcal{L}_{\mathrm{post}} + \lambda_{\mathrm{bfm}}\mathcal{L}_{\mathrm{bfm}} + \lambda_{\mathrm{img}}\mathcal{L}_{\mathrm{img}} + \lambda_{\mathrm{perc}}\mathcal{L}_{\mathrm{perc}},
    \label{eq:full_loss}
\end{equation}
where the posterior estimation loss is $\mathcal{L}_{\mathrm{post}} = \mathbb{E}[\|x_c - x\|_1 + \lambda_{\mathrm{perc}}^{(p)}\,d_{\mathrm{perc}}(x_c,x)]$. The remaining terms $\mathcal{L}_{\mathrm{img}}=\|\hat{x}-x\|_1$ and $\mathcal{L}_{\mathrm{perc}}=\|f_{\mathrm{perc}}(\hat{x})-f_{\mathrm{perc}}(x)\|_2^2$ enforce image fidelity and perceptual quality over the noisy channel. The complete channel-inclusive training loop is summarized in Algorithm~\ref{alg:train}. At inference time, the receiver simply computes $x_c=D_{\mathrm{pos}}(\bm{y},\xi)$ and solves the empirical ODE using $K$ Euler steps to yield the final reconstruction $\hat{x}=x_K$.

\begin{algorithm}[t]
    \caption{Two-Stage Training of RC-BFM}
    \label{alg:train}
    \begin{algorithmic}[1]
        \REQUIRE Image batch $\{x_i\}_{i=1}^B$, channel sampler
        \ENSURE Trained $E_{\mathrm{tx}}$, $D_{\mathrm{pos}}$, $B_\theta$
        \STATE \textit{/* Stage~1: channel-free pre-training of $B_\theta$ */}
        \STATE Form pseudo $x_c$ by perturbing $x_i$; draw $t\sim\mathcal{U}[\delta,1{-}\delta]$; build the bridge state via~\eqref{eq:bridge_path}; minimize $\mathcal{L}_{\mathrm{bfm}}$ to update $B_\theta$.
        \STATE \textit{/* Stage~2: channel-aware joint training */}
        \STATE Encode $\bm{s}_i{=}E_{\mathrm{tx}}(x_i)$; sample the channel and receive $\bm{y}_i$; estimate $x_{c,i}{=}D_{\mathrm{pos}}(\bm{y}_i,\xi_i)$; compute $\mathcal{L}_{\mathrm{post}}$.
        \STATE Build $C_{ij}{=}\|x_{c,i}-x_j\|_2^2$; solve RC-OT~\eqref{eq:ot_sb}; draw endpoint pairs from $\pi^*$.
        \STATE Sample $t\sim\mathcal{U}[\delta,1{-}\delta]$; construct $x_t$ via~\eqref{eq:bridge_path}; compute $\mathcal{L}_{\mathrm{bfm}},\mathcal{L}_{\mathrm{img}},\mathcal{L}_{\mathrm{perc}}$; update $E_{\mathrm{tx}},D_{\mathrm{pos}},B_\theta$ via~\eqref{eq:full_loss}.
    \end{algorithmic}
\end{algorithm}

\begin{figure*}[t]
    \centering
    \includegraphics[width=\linewidth]{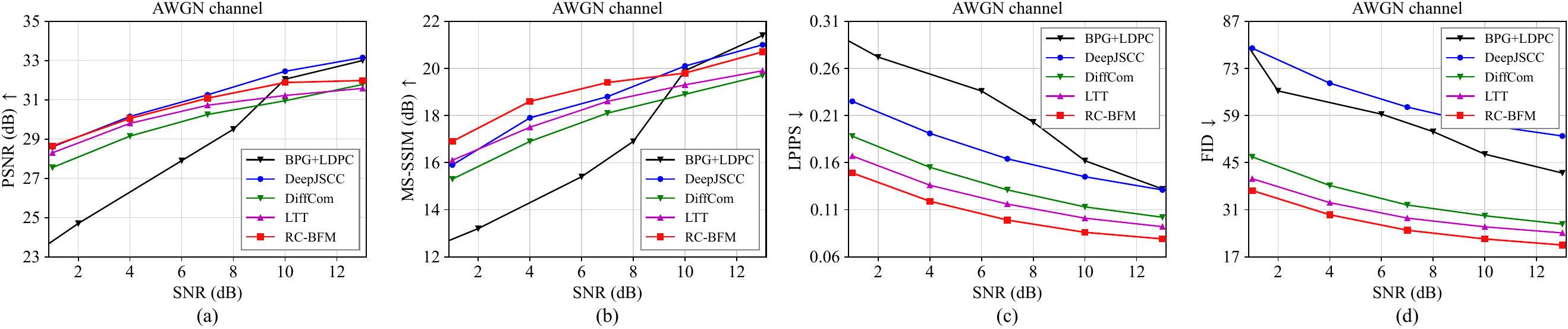}
    \caption{(a) PSNR, (b) MS-SSIM, (c) LPIPS, and (d) FID performance versus SNR of FFHQ-64$\times$64 with AWGN channel and CBR $\rho{=}1/4$.}
    \label{fig:ffhq_snr}
\end{figure*}

\section{Experimental Results}
\subsection{Simulation Settings}

We evaluate on CIFAR-10~\cite{krizhevsky2009learning} at $32{\times}32$ and FFHQ at $64{\times}64$ resolutions. Fidelity is measured by PSNR and MS-SSIM, and perceptual quality by LPIPS~\cite{zhang2018unreasonable} and FID~\cite{heusel2017gans}. For consistency with the curve plots, we report MS-SSIM in dB. We compare with \textit{BPG + 5G LDPC}~\cite{richardson2018design}, discriminative receiver \textit{DeepJSCC}~\cite{bourtsoulatze2019deep}, diffusion-based receiver \textit{DiffCom}~\cite{wang2025diffcom}, and FM-based receiver \textit{LTT}~\cite{fu2026ltt}. All modules use SongUNet~\cite{Song2021Score} with 67.3\,M parameters. Training uses AdamW with $\mathrm{lr}=10^{-4}$, cosine decay, mixed SNR $\gamma\sim\mathcal{U}[1,13]$\,dB, and CBR $\rho\in\{1/12,1/6,1/4,5/12,1/2\}$. Unless otherwise stated, we use $\sigma{=}0.5$, $\varepsilon{=}0.1$, $\alpha{=}0.5$, $\delta{=}0.01$, 50 Sinkhorn iterations, and Euler sampling with $4$ NFEs. Latency is measured per image on a single NVIDIA A800 GPU at batch size~1.

\subsection{Main Comparison}

\begin{table*}[t]
    \centering
    \caption{Main comparison under AWGN at $10$\,dB, $\rho{=}1/4$. Latency (ms) measured per image on a single A800 GPU. Best results in bold.}
    \label{tab:main_awgn}
    \begingroup
    \renewcommand{\arraystretch}{1.05}
    \begin{tabular}{llcccccc}
        \toprule
        \multirow{2}{*}{Dataset} & \multirow{2}{*}{Method}               & \multicolumn{2}{c}{Fidelity} & \multicolumn{2}{c}{Perception} & \multicolumn{2}{c}{Efficiency}                                                                   \\
        \cmidrule(lr){3-4} \cmidrule(lr){5-6} \cmidrule(lr){7-8}
                                 &                                       & PSNR $\uparrow$              & MS-SSIM (dB) $\uparrow$        & LPIPS $\downarrow$             & FID $\downarrow$ & NFE $\downarrow$ & Latency (ms) $\downarrow$ \\
        \midrule
        \multirow{5}{*}{CIFAR-10}
                                 & BPG\,+\,5G LDPC                       & 27.62                        & 19.5                            & 0.210                          & 80.2             & --               & 189.6                     \\
                                 & DeepJSCC~\cite{bourtsoulatze2019deep} & \textbf{31.27}               & \textbf{21.8}                           & 0.146                          & 92.1             & \textbf{1}       & \textbf{79.6}             \\
                                 & DiffCom~\cite{wang2025diffcom}        & 29.48                        & 19.4                           & 0.118                          & 31.8             & 200              & 7593                      \\
                                 & LTT~\cite{fu2026ltt}                  & 29.18                        & 18.9                           & 0.130                          & 38.1             & 10               & 258.6                     \\
                                 & \method{}                             & 30.34                        & 21.1                  & \textbf{0.098}                 & \textbf{25.6}    & 4                & 129.2                     \\
        \midrule
        \multirow{5}{*}{FFHQ-64$\times$64}
                                 & BPG\,+\,5G LDPC                       & 32.25                        & 19.9                           & 0.162                          & 47.5             & --               & 192.2                     \\
                                 & DeepJSCC~\cite{bourtsoulatze2019deep} & \textbf{32.85}               & \textbf{20.8}                  & 0.145                          & 56.2             & \textbf{1}       & \textbf{85.4}             \\
                                 & DiffCom~\cite{wang2025diffcom}        & 31.22                        & 19.3                           & 0.101                          & 25.9             & 200              & 7964                      \\
                                 & LTT~\cite{fu2026ltt}                  & 30.95                        & 18.9                           & 0.113                          & 29.2             & 10               & 319.7                     \\
                                 & \method{}                             & 31.58                        & 19.6                           & \textbf{0.086}                 & \textbf{22.3}    & 4                & 156.4                     \\
        \bottomrule
    \end{tabular}
    \endgroup
\end{table*}

Table~\ref{tab:main_awgn} shows results under AWGN at $10$\,dB and CBR $\rho{=}1/4$. We report the best metrics achieved by each model, regardless of decoding latency. DeepJSCC achieves the highest PSNR and MS-SSIM on FFHQ-64 but performs poorly on LPIPS and FID, reflecting the known limits of distortion-only training. DiffCom improves perceptual quality but requires $200$ NFEs and seconds of latency. LTT runs faster but trails \method{} in both LPIPS and FID. \method{} balances fidelity and perception: on FFHQ-64, it drops FID to $22.3$ (down from DiffCom's $25.9$ and LTT's $29.2$) using just $4$ NFEs and $156.4$\,ms, while maintaining competitive PSNR.

\begin{figure}[t]
    \centering
    \includegraphics[width=\columnwidth]{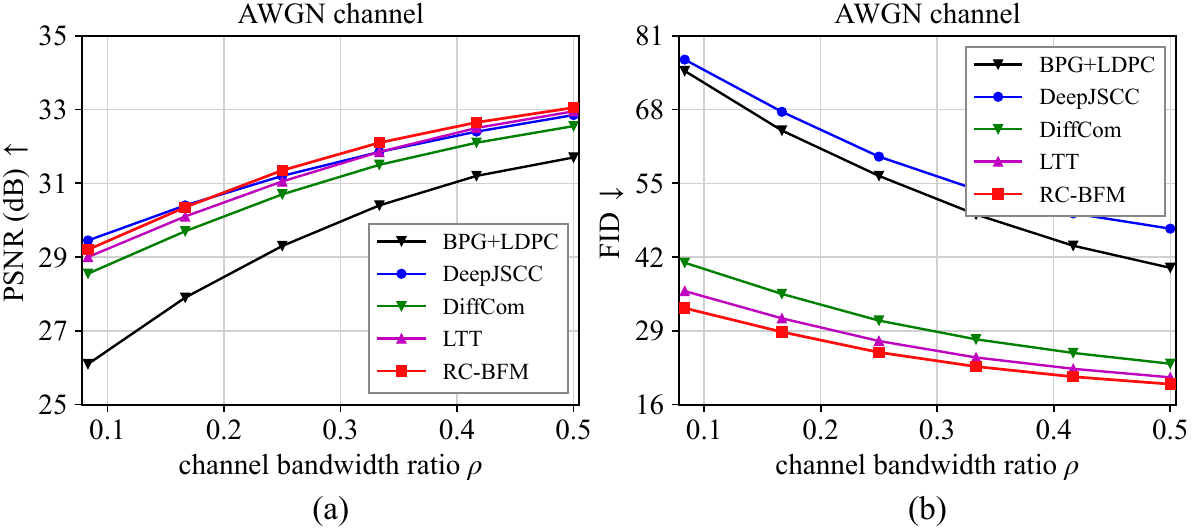}
    \caption{(a) PSNR and (b) FID versus CBR $\rho$ of FFHQ-64$\times$64 with AWGN channel at SNR $=9$\,dB.}
    \label{fig:ffhq_cbr}
\end{figure}

\begin{figure}[t]
    \centering
    \includegraphics[width=\columnwidth]{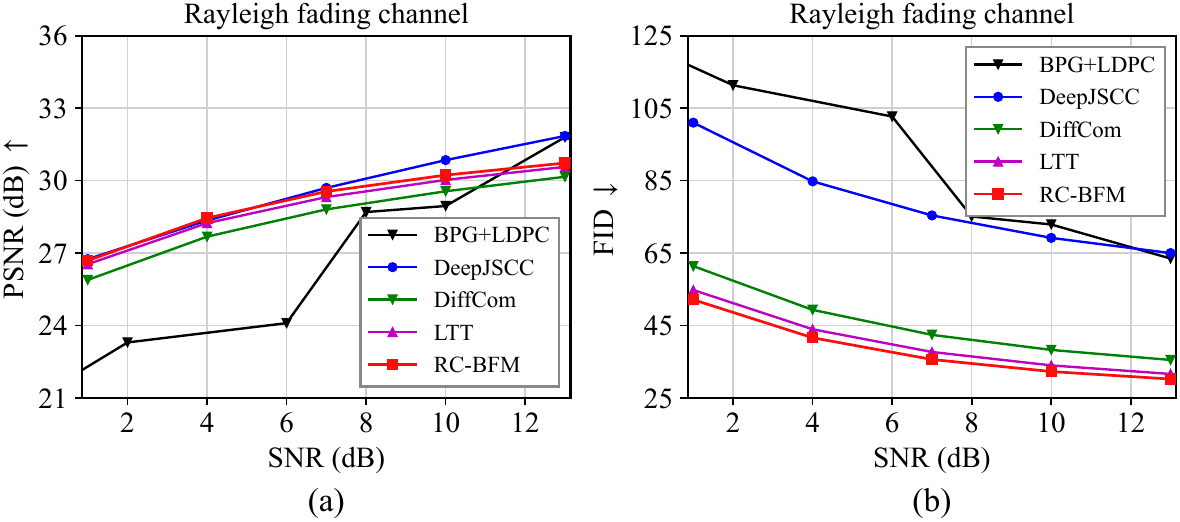}
    \caption{(a) PSNR and (b) FID versus SNR of FFHQ-64$\times$64 with Rayleigh fading channel and CBR $\rho{=}1/4$.}
    \label{fig:ffhq_rayleigh_snr}
\end{figure}

Fig.~\ref{fig:ffhq_snr} plots performance across AWGN SNRs. \method{} achieves the best LPIPS and FID at all SNRs and stays competitive in PSNR and MS-SSIM. At low SNRs, generative receivers beat BPG\,+\,LDPC even on distortion metrics. Above $10$\,dB, DeepJSCC and BPG\,+\,LDPC recover higher PSNR and MS-SSIM, exposing the trade-off between exact pixel matching and perceptual realism.

Fig.~\ref{fig:ffhq_cbr} shows the effect of available bandwidth at $9$\,dB. Performance improves rapidly at low CBRs before saturating. \method{} holds the lowest FID at all bandwidths, and its PSNR matches or exceeds the top baselines at medium-to-high CBRs. This perceptual advantage holds under Rayleigh fading (Fig.~\ref{fig:ffhq_rayleigh_snr}): \method{} consistently yields the best FID, DeepJSCC dominates PSNR, and BPG\,+\,LDPC only catches up at the highest SNR.

\subsection{Latency--Quality Trade-off}

\begin{table}[t]
    \centering
    \caption{\method{} on FFHQ-64$\times$64 under AWGN at SNR $=10$\,dB with $\rho{=}1/4$: latency--quality trade-off versus NFE. Latency (ms) is per image on one NVIDIA A800 GPU at batch size~$1$.}
    \label{tab:tradeoff_nfe}
    \begingroup
    \renewcommand{\arraystretch}{1.05}
    \begin{tabular}{c c c c c}
        \toprule
        NFE  & PSNR (dB) $\uparrow$ & LPIPS $\downarrow$ & FID $\downarrow$ & Latency (ms) $\downarrow$ \\
        \midrule
        $1$  & $29.94$              & $0.103$            & $28.4$           & $\mathbf{42.9}$           \\
        $2$  & $30.88$              & $0.094$            & $24.9$           & $82.2$                    \\
        $4$  & $\mathbf{31.58}$     & $0.086$            & $22.3$           & $156.4$                   \\
        $8$  & $31.52$              & $0.082$            & $20.8$           & $295.7$                   \\
        $16$ & $31.41$              & $0.080$            & $19.9$           & $575.4$                   \\
        $32$ & $31.26$              & $0.079$            & $19.4$           & $1136.8$                  \\
        $64$ & $31.10$              & $\mathbf{0.078}$   & $\mathbf{19.1}$  & $2270.1$                  \\
        \bottomrule
    \end{tabular}
    \endgroup
\end{table}

Table~\ref{tab:tradeoff_nfe} tracks how NFE budgets affect \method{}. FID and LPIPS improve steadily as NFE increases. PSNR, however, peaks at $4$~NFEs and then slowly declines---another instance of the fidelity--perception trade-off. We find $4$~NFEs to be a practical sweet spot: it maximizes PSNR ($31.58$\,dB) and hits an FID of $22.3$ in $156.4$\,ms. Doubling the budget to $8$~NFEs pushes FID down to $20.8$ but nearly doubles the latency.

\begin{figure*}[t]
    \centering
    \includegraphics[width=\linewidth]{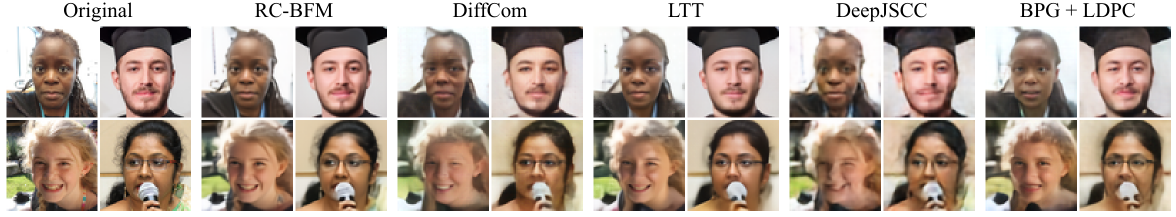}
    \caption{Qualitative comparison on FFHQ-64$\times$64 under AWGN at SNR $=10$\,dB, $\rho{=}1/4$. \method{} is evaluated at $4$~NFEs, DiffCom at $100$~NFEs, and LTT at $10$~NFEs.}
    \label{fig:ffhq_qualitative}
\end{figure*}

\begin{figure}[t]
    \centering
    \includegraphics[width=\columnwidth]{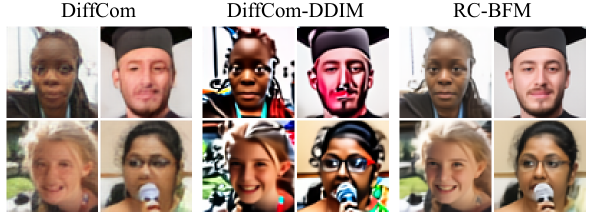}
    \caption{$10$~NFEs qualitative comparison on FFHQ-64$\times$64 under AWGN at SNR $=10$\,dB and $\rho{=}1/4$.}
    \label{fig:ffhq_qualitative_same_nfe}
\end{figure}

Figs.~\ref{fig:ffhq_qualitative} and~\ref{fig:ffhq_qualitative_same_nfe} show these results visually. Even at its much lower default NFE budget, \method{} produces the most consistent reconstructions. Forcing all generative receivers to use exactly $10$ NFEs makes the differences obvious: DiffCom under-converges and outputs blurry faces; DiffCom-DDIM speeds up the reverse chain but hallucinates facial structures and saturates colors. \method{} preserves both the global pose and fine details. This happens because the realization-coupled bridge initializes near the target, making few-step decoding naturally easy---we do not have to force an aggressive solver along a long noise-to-data path.

\subsection{Coupling Ablation}

\begin{table}[t]
    \centering
    \caption{Coupling ablation on FFHQ-64$\times$64, AWGN, $10$\,dB, $\rho{=}1/4$. Cost is the normalized minibatch matching cost; straightness is the cosine similarity between velocity and endpoint displacement; NFE$^*$ is the smallest NFE reaching FID $\le 25$.}
    \label{tab:ablation_ot}
    \resizebox{\linewidth}{!}{
        \begin{tabular}{lccccc}
            \toprule
            Coupling                     & PSNR $\uparrow$ & FID $\downarrow$ & Cost $\downarrow$ & Straight.\ $\uparrow$ & NFE$^*$ $\downarrow$ \\
            \midrule
            Independent coupling         & 29.98           & 31.4             & 0.842             & 0.624                 & 20                   \\
            Minibatch OT                 & 30.72           & 27.8             & 0.587             & 0.742                 & 12                   \\
            Entropic OT ($\alpha{=}0$)   & 31.10           & 25.4             & 0.534             & 0.798                 & 8                    \\
            RC-OT ($\alpha{=}0.5$, ours) & \textbf{31.58}  & \textbf{22.3}    & \textbf{0.421}    & \textbf{0.917}        & \textbf{4}           \\
            \bottomrule
        \end{tabular}
    }
\end{table}

Table~\ref{tab:ablation_ot} compares four coupling strategies. Independent coupling creates highly curved trajectories with the highest matching cost; it needs $20$ NFEs to break an FID of $25$. Introducing optimal transport and entropic regularization straightens the paths and improves reconstructions. RC-OT performs best across all metrics: it achieves the lowest matching cost ($0.421$), the straightest paths ($0.917$), and hits the target FID $\le 25$ in just $4$ NFEs. As argued in Section~III-B, enforcing the physical source--channel pairing makes decoding both faster and more accurate.

\section{Conclusion}

This paper identifies a coupling mismatch in FM receivers for semantic communication: independent endpoint sampling discards the physical source--channel pairing imposed by the channel realization, inducing a conditional train--test distribution shift and inflating trajectory curvature. Restricting training pairs to the same channel realization removes the shift by construction, which we instantiate as a realization-coupled entropic OT solution. Embedded into a bridge FM decoder operating in pixel space, the method achieves the best perceptual quality at the reported $4$-NFE operating point on CIFAR-10 and FFHQ-64$\times$64 under AWGN and Rayleigh fading. Future work will extend the framework to high-resolution latent-space decoding and end-to-end latency optimization.

\balance
\bibliographystyle{IEEEtran}
\bibliography{IEEEabrv,ref}

\end{document}